\documentclass[oneside,english]{amsart}
\usepackage[T1]{fontenc}
\usepackage[latin9]{inputenc}
\usepackage{xcolor}
\usepackage{amsthm}
\usepackage{amstext}
\usepackage{amssymb}
\usepackage{esint}
\PassOptionsToPackage{normalem}{ulem}
\usepackage{ulem}

\makeatletter

\providecolor{lyxadded}{rgb}{0,0,1}
\providecolor{lyxdeleted}{rgb}{1,0,0}

\numberwithin{equation}{section}
\numberwithin{figure}{section}

\usepackage[T1]{fontenc}
\usepackage[latin9]{inputenc}
\usepackage{amsthm}
\usepackage{amstext}
\usepackage{amssymb}
\usepackage{esint}

\numberwithin{equation}{section}
\numberwithin{figure}{section}
  \theoremstyle{plain}
  \newtheorem{prop}{\protect\propositionname}
  \theoremstyle{plain}
  \newtheorem{assumption}{\protect\assumptionname}
  \theoremstyle{plain}
  \newtheorem{lem}{\protect\lemmaname}
  \theoremstyle{remark}
  \newtheorem{rem}{\protect\remarkname}
 \theoremstyle{definition}
\newtheorem{exmp}{Example}[section]
\newtheorem{cor}{Corollary}

\theoremstyle{plain}
\newtheorem{thm}{\protect\theoremname}

\makeatother

\usepackage{babel}
  \providecommand{\assumptionname}{Assumption}
  
  \providecommand{\lemmaname}{Lemma}
  \providecommand{\propositionname}{Proposition}
  \providecommand{\remarkname}{Remark}
\providecommand{\theoremname}{Theorem}

\makeatother

\usepackage{babel}
\begin{document}

\title{Multiscale time averaging, Reloaded}

\author{Shmuel Fishman}

\address{Physics Department, Technion, Haifa 32000, Israel}

\author{Avy Soffer}

\address{Mathematics Department, Rutgers University, 110 Frelinghuysen road,
Piscataway , NJ 08854, USA}

\email{soffer@math.rutgers.edu}

\maketitle

\begin{abstract}
We develop a rigorously controlled multi-time scale averaging technique;
the averaging is done on a finite time interval, properly chosen,
and then, via iterations and normal form transformations, the time
intervals are scaled to arbitrary order. Here, we consider as an example
the problem of a finite dimensional conservative dynamical system,
which is quasiperiodic and dominated by slow frequencies, leading
to small divisor problems in perturbative schemes. Our estimates hold
for arbitrary long time intervals similar to the Nekhoroshev type
results.
\end{abstract}

\section{\label{sec:Introduction}Introduction}

The aim of this work is to initiate the multiscale time-averaging
analysis of repeated normal form finite time interval averages. That
is we study the problem of general dynamical system type, written
in the form
\begin{equation}
i\frac{\partial}{\partial t}\vec{c}=\beta A(t)\vec{c}\label{eq:1.1}
\end{equation}
 where $\beta$ is a small parameter and $\vec{c}$ is restricted
to finite dimension. We choose $A(t)$ to be a hermitian matrix, for
each time $t$, that is we restrict ourselves to conservative systems
, i.e.,
\begin{equation}
\|\vec{c}(t)\|^{2}=\sum_{i=1}^{N}\left|c_{i}(t)\right|^{2}=\left\Vert \vec{c}(0)\right\Vert ^{2}.
\end{equation}
 where $N$ is the dimension of space. The models we consider here
are directly related to the problem of solving nonlinear dynamical
systems they are modeled by various types of equations (ODE and PDE)
{[}see an example in section 4{]}. Standard perturbation expansion
results in three difficulties in general : the secular terms, the
small divisor problem and the entropy problem. The secular terms are
divergent terms . These terms are removed by an almost identity transformation
of the original variables . This approach was introduced by Poincaré,
and the transforms are called normal form transformations.

It was shown in \cite{Ziane:2000:CRG,DeVille20081029,OMalleyJr20063}
that the transformation group method of Oono \cite{PhysRevE.54.376}
can be applied to ODE's and results in a similar normal form transformation.
In \cite{Soffer1999,Soffer2004} a new approach was developed to normal
form transformations which applies equally well to PDE's. It relies
on the averaging and it is particularly suitable for the analysis
of the present paper {[}See proof of theorems 1-2{]}.

The normal form analysis does not solve the small divisor problem
; these are perturbative terms which are not infinite but with arbitrary
small denominators.

The way to deal with those terms is usually based on KAM theory \cite{Arnold1989},
where one requires the numerators to be exponentially small and by
imposing diophantine conditions on the initial data.This leads to
a polynomial lower bound on the small divisors . This approach can
not be used in general for explicitly given systems and initial data.

Another approach is the method of Nekhoroshev {[}see \cite{2009arXiv0912.3725B}
and the references therin{]} in this approach one proves estimates
to times of order $e^{-\frac{c}{\beta^{\alpha}}}$ with $\beta$ small,
and without much restrictions on the system.

Here we develop a new approach based on partial time averaging repeated
on larger and larger scales, to deal with the small divisor problem,
without resorting to diophantine conditions and to get Nekhoroshev
type results. The key observation is that when the denominator is
very small compared with the inverse time scale of averaging, it can
be approximated by zero(!), and then one removes this term by a normal
form transformation just as secular terms.

The entropy problem is encountered when approximating PDE's. In this
case the number of terms after few iterations is astronomically large{[}
See \cite{1367-2630-12-6-063035,Fishman2008a,Fishman2009a,0951-7715-25-4-R53}{]}.

Our approach extends the time averaging to arbitrary time intervals
, in contrast to standard time averaging analysis as presented in
\textbf{\cite{Sanders2010,Krol1991}.}

We construct finite time interval averages of $A(t)$
\begin{equation}
\bar{A}_{0}^{\left(n\right)}=\frac{1}{T_{0}}\int_{nT_{0}}^{(n+1)T_{0}}A(s)ds
\end{equation}
 as approximation of the dynamics on a finite time interval depending
on $T_{0}$. The equation is solved for the piecewise constant evolution
and we peel off the approximate solution from the exact one and derive
an equation for the \textquotedbl{}leftover\textquotedbl{}.\\
Then, we introduce a normal form transformation, based on the method
of \cite{Soffer1999,Soffer2004}; this is an almost identity transformation
of the \textquotedbl{}leftover\textquotedbl{}. We then show, following
(properly modified) arguments of \cite{Soffer1999,Soffer2004}, that
the resulting equation for the new quantity satisfies the same equation
as (\ref{eq:1.1}), but with the replacement
\begin{equation}
\beta\rightarrow\beta^{3/2}\label{eq:1.4}
\end{equation}

\begin{equation}
A(t)\rightarrow\tilde{A}(t),
\end{equation}
 with $\tilde{A}(t)$ containing only \textquotedbl{}low frequency\textquotedbl{}
terms. By repeating the above iteration we get, to any order, an operator
$V_{n}(t)$, which solves the original equation for the evolution
operator to order $\beta^{3n/2}$, namely, the almost constant leftover,
$\vec{c}_{M}$ satisfies,
\begin{equation}
\vec{c}(t)=V_{1}V_{2}\cdot\cdot\cdot\cdot\cdot V_{M-1}\vec{c}_{M}(t),
\end{equation}
 with
\begin{equation}
V_{n}=U_{1,n}U_{2,n}\tilde{U}_{n}^{-1},\label{eq:1.7}
\end{equation}
 where $U_{1,n}$ and $U_{2,n}$ are unitary ``peel off'' transformations
while $\tilde{U}_{n}^{-1}$ is a normal form transformation. The normal
form transformation is an almost identity transformation,
\begin{equation}
\tilde{U}_{n}=I+O\left(\beta^{3n/2}\right).
\end{equation}
 The $\vec{c}_{M}\left(t\right)$ satisfy the equation similar to
(\ref{eq:1.1})
\begin{equation}
i\frac{\partial}{\partial t}\vec{c}_{M}\left(t\right)=\beta^{3M/2}A_{M}\left(t\right)\vec{c}_{M}\left(t\right).
\end{equation}
 For ensuring that the new $A_{M}$ are hermitian after each iteration,
we need to modify the normal form structure used in \cite{Soffer1999,Soffer2004},
and apply the following : \begin{prop} \label{prop:1}Let, for each
$t\ge0$ the vector family $\vec{c}(t)\in\mathbb{C}^{N}$ satisfy
the following:
\begin{equation}
\|\vec{c}(t)\|=1\,\,\,\text{for any }\,\,\,\vec{c}\left(0\right)\in\mathbb{C}^{N}
\end{equation}

\begin{equation}
\vec{c}(t)=U(t)\vec{c}(0),
\end{equation}
therefore we have
\begin{equation}
U(t)^{\dagger}=U(t)^{-1},
\end{equation}
with $U$ linear. Furthermore assume
\begin{equation}
\left\Vert \frac{d\vec{c}}{dt}\right\Vert <\infty.
\end{equation}
\end{prop} Then,
\begin{equation}
i\frac{dU(t)}{dt}=A(t)U(t),
\end{equation}
 with
\begin{equation}
A^{\dagger}=A.
\end{equation}

\begin{proof} From
\begin{equation}
i\frac{d\vec{c}}{dt}=i\frac{dU}{dt}\vec{c}\left(0\right)=A(t)\vec{c}(t),\label{eq:1.15}
\end{equation}
 and the conjugate equation, we get that for arbitrary $c_{1}$ and
$c_{2}$ satisfying (\ref{eq:1.15}), the following identity
\begin{equation}
0=i\frac{d}{dt}\vec{c_{1}}^{*}\cdot\vec{c}_{2}=\vec{c}_{1}\cdot A(t)\vec{c}_{2}-A(t)\vec{c}_{1}\cdot\vec{c}_{2}=\left(\vec{c}_{1},A\vec{c}_{2}\right)-\left(A\vec{c}_{1},\vec{c}_{2}\right)=\left(\vec{c}_{1},\left(A-A^{\dagger}\right)\vec{c}_{2}\right).
\end{equation}
 \end{proof}

\begin{rem} The generalization to an infinite dimensional space requires
a different approach. \end{rem}

The method of the present work involves an infinite hierarchy of averaging
approximations resulting in the effective reduction of $\beta$ via
(\ref{eq:1.4}), allowing the validity of the approximation over increasingly
longer time-scales, $T_{0}$. This is a substantial improvement compared
to known methods where averaging in one stage is used as described
in \cite{Sanders2010} (see p.379, App. E 2.2 and p. 390) and \cite{Krol1991}.
The present work is an important step in this direction proposed in
\cite{Sanders2010}. Finally, in sec 4, we demonstrate how the method
works for the case of almost periodic matrix $A(t)$:
\begin{equation}
A(t)=\sum_{j=1}^{\infty}A_{j}e^{i\omega_{j}t}+h.c.
\end{equation}
 with
\begin{equation}
\|A_{j}\|\leq c\left\langle j\right\rangle ^{-\sigma}\quad\text{ , for some, }\quad\sigma>1
\end{equation}
 and the interesting case $\omega_{j}\rightarrow0,$ as $j\rightarrow\infty$.
Here h.c stands for hermitian conjugate. Notice that no assumptions
are made on the $\omega_{j}$, as they approach zero thus removing
the small divisor problem. Here $A_{j}$ are constant $N\times N$
matrices. The interest in these kind of systems stems from the fact
that many dynamical systems can be modeled and/or approximated by
such equations. A related, but different example is the nonlinear
system derived from the Nonlinear Schrödinger Equation (NLSE) with
a random potential term \cite{Fishman2008a,0951-7715-22-12-004}.
It leads naturally to a system with
\begin{equation}
A_{n,m}\left(t\right)=\sum_{i,j}V_{n}^{mjk}e^{i\omega_{mnjk}t}+h.c,
\end{equation}
 where the regime
\begin{equation}
0<\beta\ll1,
\end{equation}
 is of great interest \cite{0951-7715-25-4-R53}. More generally,
Hamiltonian dynamical systems with Hamiltonians of the form
\begin{equation}
H=H_{0}+\beta H_{1}\label{eq:1.22}
\end{equation}
 can be studied by solving exactly for the dynamics generated by $H_{0}$
and (\ref{eq:1.22}) is reduced to a system similar to (\ref{eq:1.1})
using the interaction picture.

Another class of examples are slowly changing (in time) interactions
$H_{1}=H_{1}\left(\beta t\right)$.

\section{\label{sec:Averaging}Averaging}

In this section an averaging of the matrix $A\left(t\right)$ will
be introduced. This averaging can be performed successively. Equation
(\ref{eq:1.1}) is replaced to a hierarchical set of equations where
$A\left(t\right)$ and $\vec{c}\left(t\right)$ are replaced by $A_{n}$
and $\vec{c}_{n}$. The starting point of the hierarchy is $A_{0}\equiv A$
and $\vec{c}_{0}\equiv\vec{c}$ satisfying the basic equation
\begin{equation}
i\frac{\partial}{\partial t}\vec{c}=\beta A(t)\vec{c}\label{eq:2.1}
\end{equation}
 where $\vec{c}=(c^{(1)}...c^{(N)})$ is a vector and $A\left(t\right)$
is a matrix. We assume the following: \begin{assumption} \label{assump:1}$A(t)$
is, for each $t$, a hermitian $N\times N$ matrix, satisfying,
\begin{equation}
\sup_{t}\left\Vert A(t)\right\Vert <M<\infty.\label{eq:2.2}
\end{equation}
 \end{assumption} Under Assumption \ref{assump:1}, the solution
of (\ref{eq:1.1}) exists and is bounded uniformly in the usual vector
norm on $\mathbb{C}^{N}$ in time, since, due to the hermitian property
of $A(t)$,
\begin{equation}
\vec{c}(t)=\mathcal{T}\, e^{-i\beta\int_{0}^{t}A(s)ds}\vec{c}(0).\label{eq:2.3}
\end{equation}
Recall that
\[
\vec{c}\left(t\right)\equiv U\left(t\right)\vec{c}\left(0\right)=\vec{c}\left(0\right)+\sum_{n=1}^{\infty}\left(-i\right)^{n}\int_{0}^{t}\int_{0}^{t_{1}}\cdots\int_{0}^{t_{n-1}}A\left(t_{1}\right)\cdots A\left(t_{n}\right)\vec{c}\left(0\right)dt_{n}\cdots dt_{1}
\]
has a constant norm. $\mathcal{T}$ stands for the time ordering (see
\cite{reed1980methods} section 2) . Moreover, by taking derivatives,
we get:
\begin{equation}
\sup_{t}\left\Vert \vec{c}'(t)\right\Vert =\sup_{t}\left\Vert \frac{\beta}{i}A(t)\vec{c}(t)\right\Vert \leq\beta\sup_{t}\left\Vert A(t)\vec{c}(t)\right\Vert \leq\beta M
\end{equation}
where $\prime$ denote the derivative with respect to $t$. $\beta$
is a small positive number $0<\beta\ll1$. \\
We introduce now averaging over intervals of size $T_{0}$. Define
the average matrix
\begin{equation}
\bar{A}_{0}^{(n)}=\frac{1}{T_{0}}\int_{nT_{0}}^{(n+1)T_{0}}A(s)ds.\label{eq:2.5}
\end{equation}
 At a latter stage the relation between $T_{0}$ and $\beta$ will
be defined. The global averaged matrix is defined by

\begin{equation}
\bar{A}_{0}^{g}(t)=\bar{A}_{0}^{(n)}\qquad\text{for}\qquad nT_{0}\leq t<T_{0}(n+1).\label{eq:2.6}
\end{equation}
 Now one defines the propagator related to the averaged $A$,
\begin{equation}
U_{0}(t)=e^{-i\beta\bar{A}_{0}^{(n)}\left(t-nT_{0}\right)}\cdots e^{-i\beta\bar{A}_{0}^{(0)}T_{0}}\qquad\text{for}\qquad nT_{0}\leq t\leq(n+1)T_{0},\label{eq:2.7}
\end{equation}
 and its inverse

\begin{equation}
U_{0}^{-1}(t)=e^{+i\beta\bar{A}_{0}^{(0)}T_{0}}e^{+i\beta\bar{A}_{0}^{(1)}T_{0}}\cdots e^{i\beta\bar{A}_{0}^{(n)}(t-nT_{0})}.\label{eq:2.8}
\end{equation}
 It is easily verified that:
\begin{equation}
U_{0}(t)U_{0}^{-1}(t)=1\qquad U_{0}^{-1}(t)U_{0}(t)=1.
\end{equation}
 By direct differentiation one finds:

\begin{equation}
i\frac{\partial}{\partial t}U_{0}(t)=\beta\bar{A}_{0}^{(n)}U_{0}(t)=\beta\bar{A}_{0}^{g}(t)U_{0}(t),\label{eq:2.10}
\end{equation}
 while the inverse satisfies
\begin{equation}
i\frac{\partial}{\partial t}U_{0}^{-1}\left(t\right)=U_{0}^{-1}\left(-\beta\bar{A}_{0}^{(n)}\right)=-\beta U_{0}^{-1}\bar{A}_{0}^{g}(t).\label{eq:2.11}
\end{equation}
 Therefore, $U_{0}$ is the propagator of the partially averaged motion
generated by $\bar{A}_{0}^{g}\left(t\right)$. The corresponding solution
of the averaged equation is,
\begin{equation}
\vec{c}_{1}\left(t\right)=U_{0}^{-1}(t)\vec{c}(t),\label{eq:2.12}
\end{equation}
as we demonstrate in what follows. It satisfies
\begin{eqnarray}
i\frac{\partial}{\partial t}\vec{c}_{1} & = & \left[i\frac{\partial}{\partial t}U_{0}^{-1}(t)\right]c(t)+U_{0}^{-1}(t)\left[i\frac{\partial}{\partial t}\vec{c}(t)\right]\label{eq:2.13}\\
 & = & \left[-\beta U_{0}^{-1}\bar{A}_{0}^{g}(t)\right]c(t)+U_{0}^{-1}\left[\beta A(t)\vec{c}(t)\right],\nonumber
\end{eqnarray}
 where (\ref{eq:2.10}) and (\ref{eq:2.11}) were used. Using (\ref{eq:2.12})
one gets
\begin{equation}
i\frac{\partial}{\partial t}\vec{c}_{1}(t)=-\beta U_{0}^{-1}\bar{A}_{0}^{g}(t)U_{0}(t)\vec{c}_{1}(t)+\beta U_{0}^{-1}A(t)U_{0}(t)\vec{c}_{1}(t).
\end{equation}
 Hence,
\begin{equation}
i\frac{\partial}{\partial t}\vec{c}_{1}(t)=\beta U_{0}^{-1}\left[A(t)-\bar{A}_{0}^{g}(t)\right]U_{0}(t)\vec{c}_{1}(t).
\end{equation}
 This equation is analogous to (\ref{eq:1.1}), if written in the
form
\begin{equation}
i\frac{\partial}{\partial t}\vec{c}_{1}(t)=\beta A_{1}(t)\vec{c}_{1}(t),\label{eq:2.16}
\end{equation}
 with
\begin{equation}
A_{1}(t)=U_{0}^{-1}\left[A(t)-\bar{A}_{0}^{g}(t)\right]U_{0}(t).\label{eq:2.17}
\end{equation}
 Performing on (\ref{eq:2.16}) operations, similar to the ones performed
on (\ref{eq:2.1}), leads to the next stage of the hierarchy. Before
doing that we analyze the meaning of the transformation to (\ref{eq:2.16}).
The dynamics generated by $A_{1}(t),$ is the peeling of $A(t)$ by
the average dynamics
\begin{equation}
\bar{A}_{0}^{g}(t)=\bar{A}_{0}^{(n)}(n=[t/T_{0}])=\frac{1}{T_{0}}\int_{nT_{0}}^{(n+1)T_{0}}A(s)ds,
\end{equation}
 where
\begin{equation}
n=\left[t/T_{0}\right]
\end{equation}
 is the integer part of $t/T_{0}$. We turn to estimate the quantity
\begin{eqnarray}
\int_{0}^{t}\left(A(t^{\prime})-\bar{A}_{0}^{g}(t^{\prime})\right)dt^{\prime} & = & \int_{0}^{T_{0}}dt^{\prime}\left[A(t^{\prime})-\bar{A}_{0}^{(0)}\right]+\int_{T_{0}}^{2T_{0}}dt^{\prime}\left[A(t^{\prime})-\bar{A}_{0}^{(1)}\right]\nonumber \\
 & +\cdots+ & \int_{t-nT_{0}}^{t}dt^{\prime}\left[A(t^{\prime})-\bar{A}_{0}^{(n)}\right]
\end{eqnarray}
 Due to (\ref{eq:2.5}) and (\ref{eq:2.6}),
\begin{equation}
\int_{kT_{0}}^{(k+1)T_{0}}dt^{\prime}\left(A(t^{\prime})-\bar{A}_{0}^{g}(t^{\prime})\right))=0.\label{eq:2.21}
\end{equation}
 Therefore,
\begin{eqnarray}
\left|\int_{0}^{t}dt^{\prime}\left(A(t')-\bar{A}_{0}^{g}(t')\right)\right| & \leq\max_{t'} & \left|A(t^{\prime})-\bar{A}_{0}^{(n)}(t^{\prime})\right|(t-nT_{0})\nonumber \\
 & \leq & \max_{t'}\left|A(t^{\prime})-\bar{A}_{0}^{(n)}(t')\right|T_{0}\label{eq:2.22}
\end{eqnarray}
 where $nT_{0}\leq t^{\prime}\leq(n+1)T_{0}$. This result can be
summarized as: \begin{lem} \label{lem:1}
\begin{equation}
\left\Vert \int_{0}^{t}\left(A(t^{\prime})-\bar{A}_{0}^{g}(t^{\prime})\right)dt^{\prime}\right\Vert \leq2|||A|||_{t}T_{0}.\label{eq:2.23}
\end{equation}
 Here $|||A|||_{t}=\sup_{0\leq t^{\prime}\leq t}\|A(t^{\prime})\|$.
We used the fact that $|||\bar{A}_{0}^{(n)}(t)|||_{t}\leq|||A(t)|||_{t}.$
\end{lem} In what follows we choose

\begin{equation}
T_{0}=\frac{1}{\sqrt{\beta}}.\label{eq:2.24}
\end{equation}
 In order to estimate the variation of $A_{1}$ it is useful to estimate
the quantity:
\begin{equation}
I_{1}=\int_{0}^{t}dt^{\prime}A_{1}(t^{\prime})\label{eq:2.25}
\end{equation}
 or explicitly
\begin{equation}
I_{1}=\int_{0}^{t}U_{0}^{-1}(t^{\prime})\left[A(t^{\prime})-\bar{A}_{0}^{g}(t^{\prime})\right]U_{0}(t^{\prime})dt^{\prime}\label{eq:2.26}
\end{equation}
\begin{lem} $I_{1}$ satisfies :
\begin{equation}
I_{1}=O\left(\frac{1}{\sqrt{\beta}}\right)+O\left(\sqrt{\beta}t\right).\label{eq:2.31 formally}
\end{equation}
\begin{proof}

For this purpose we note that integration by parts for matrices $M_{i}$
is:
\begin{eqnarray}
\int_{0}^{t}M_{1}(t^{\prime})M_{2}(t^{\prime})M_{3}(t^{\prime})dt^{\prime} & = & M_{1}(t)\left[\int_{0}^{t}M_{2}(t^{\prime})dt^{\prime}\right]M_{3}(t)\nonumber \\
 & - & \int_{0}^{t}dt^{\prime}\left(\frac{d}{dt^{\prime}}M_{1}(t^{\prime})\right)\left[\int_{0}^{t^{\prime}}M_{2}(s)ds\right]M_{3}(t^{\prime})\nonumber \\
 & - & \int_{0}^{t}dt^{\prime}M_{1}(t^{\prime})\left[\int_{0}^{t}M_{2}(s)ds\right]\left[\frac{dM_{3}(t^{\prime})}{dt^{\prime}}\right]\label{eq:2.27}
\end{eqnarray}
 Therefore, using (\ref{eq:2.26}), (\ref{eq:2.10}) and (\ref{eq:2.11}),
\begin{eqnarray}
I_{1} & = & U_{0}^{-1}(t)\left(\int_{0}^{t}\left[A(t^{\prime})-\bar{A}_{0}^{g}(t^{\prime})\right]dt^{\prime}\right)U_{0}(t)\label{eq:I1}\\
 & - & \int_{0}^{t}\left(\beta U_{0}^{-1}\bar{A}_{0}^{g}\left(t^{\prime}\right)\right)\left(\int_{0}^{t^{\prime}}\left[A(s)-\bar{A}_{0}^{g}(s)\right]ds\right)U_{0}(t^{\prime})dt^{\prime}\nonumber \\
 & - & \int_{0}^{t}U_{0}^{-1}(t^{\prime})\left(\int_{0}^{t^{\prime}}\left[A(s)-\bar{A}_{0}^{g}(s)\right]d't\right)\left(-i\beta\bar{A}_{0}^{g}(t^{\prime})\right)U_{0}(t^{\prime})dt^{\prime}.\nonumber
\end{eqnarray}
 Using (\ref{eq:2.23}) and (\ref{eq:2.24}) one finds that the first
term is of order,
\begin{equation}
\frac{1}{\sqrt{\beta}}
\end{equation}
 while the other two terms are of order,
\begin{equation}
\beta\frac{1}{\sqrt{\beta}}t=\sqrt{\beta}t.
\end{equation}
the result (\ref{eq:2.31 formally}) follows, using the fact that
$\bar{A_{0}^{g}}=O\left(1\right)$. \end{proof}

\end{lem} We are now ready to state the first main theorem

\begin{thm}

\label{-(Time-Averaging:thoerm1}(Time Averaging for one step )

\end{thm}

Let $\vec{c}\left(t\right)\in\mathbb{C}^{N}\,,\, N<\infty$ satisfy
\begin{equation}
i\frac{\partial}{\partial t}\vec{c}\left(t\right)=\beta A\left(t\right)\vec{c}\left(t\right)
\end{equation}
where $A\left(t\right)$ are hermitian $N\times N$ matrices , for
each $t$ and $0<\beta\ll1$.

Assume moreover that
\begin{equation}
\sup_{t}\left\Vert A\left(t\right)\right\Vert <M<\infty.
\end{equation}
Then, for a partially time averaged unitary flow $U_{0}\left(t\right)$:
\begin{equation}
\vec{c}\left(t\right)=U_{0}\left(t\right)\vec{c}_{1}\left(t\right)
\end{equation}
and $c_{1}\left(t\right)$ satisfies the following estimate on any
time interval $T_{i}\leq t<T_{f}$:
\begin{equation}
\sup_{T_{i}\leq t<T_{f}}\left\Vert \vec{c}_{1}\left(t\right)-\vec{c}_{1}\left(T_{i}\right)\right\Vert \leq C\left(\sqrt{\beta}+\beta^{\frac{3}{2}}\left|T_{f}-T_{i}\right|+\beta^{\frac{5}{2}}\left|T_{f}-T_{i}\right|^{2}\right).
\end{equation}
 In particular for $\left|T_{f}-T_{i}\right|<\frac{1}{\beta}$ :
\begin{equation}
\sup_{T_{i}\leq t<T_{f}}\left\Vert \vec{c}_{1}\left(t\right)-\vec{c}_{1}\left(T_{i}\right)\right\Vert \leq C\beta^{\frac{1}{2}}.
\end{equation}

\begin{proof}

From Eq. \ref{eq:2.16} , we have that
\begin{equation}
i\left[\vec{c}_{1}\left(T_{f}\right)-\vec{c}_{1}\left(T_{i}\right)\right]=\int_{T_{i}}^{T_{f}}i\partial_{t}\vec{c}_{1}\left(t\right)dt=\beta\int_{T_{i}}^{T_{f}}A_{1}\left(t\right)\vec{c}_{1}\left(t\right)dt.
\end{equation}
 Integration by parts gives
\begin{equation}
\vec{c}_{1}\left(T_{f}\right)-\vec{c}_{1}\left(T_{i}\right)=\vec{b}_{1}+\vec{b}_{2},\label{eq:2.33-1}
\end{equation}
where
\[
b_{1}=\left.\frac{\beta}{i}\left(\int_{T_{i}}^{s}A_{1}\left(s'\right)ds'\right)\vec{c}_{1}\left(s\right)\right|_{s=T_{i}}^{T_{f}}
\]
And
\[
\vec{b}_{2}=-\frac{\beta}{i}\int_{T_{i}}^{T_{f}}\left(\int_{T_{i}}^{s}A_{1}\left(s'\right)ds'\right)\frac{\partial\vec{c}_{1}\left(s\right)}{\partial s}ds.
\]
By our estimates on $A_{1},$
\begin{equation}
\left\Vert \vec{b}_{1}\right\Vert =\left\Vert \beta\int_{T_{i}}^{T_{f}}A_{1}\left(s'\right)ds'\vec{c}_{1}\left(s\right)\right\Vert \leq C\beta\left\Vert \vec{c}_{1}\left(T_{f}\right)\right\Vert \left\Vert \int_{T_{i}}^{T_{f}}A_{1}\left(s'\right)ds'\right\Vert \leq C\beta\left(\frac{1}{\sqrt{\beta}}+\sqrt{\beta}\left|T_{f}-T_{i}\right|\right),
\end{equation}
by lemma 2. Similarly
\begin{align}
\left\Vert \vec{b}_{2}\right\Vert  & \leq C\beta\int_{T_{i}}^{T_{f}}ds\left(\frac{1}{\sqrt{\beta}}+\sqrt{\beta}s\right)\beta\left\Vert A\left(s\right)\right\Vert \nonumber \\
\leq & C\beta^{2}\left|T_{f}-T_{i}\right|\left(\frac{1}{\sqrt{\beta}}+\sqrt{\beta}\left|T_{f}-T_{i}\right|\right)\nonumber \\
= & C\beta^{\frac{3}{2}}\left(\left|T_{f}-T_{i}\right|+\beta\left|T_{f}-T_{i}\right|^{2}\right).\label{eq:2.35}
\end{align}
From these bounds the statement of the theorem follows \end{proof}

\begin{rem}

As seen from the proof , integration by parts shows that there are
two contributions to $\vec{c_{1}}\left(t\right)$ that are of different
type. The first contribution is from the boundary term and is of order
$\beta^{\frac{1}{2}}$ for all times, while the second term is smaller
for times less than $\frac{1}{\beta}$.

Furthermore the first term only depends on $\vec{c}_{1}\left(T_{f}\right)$
and not on intermediate times.

\end{rem}

\begin{rem}

The above observation about the structure of the first term allows
one to redefine $\vec{c}_{1}\left(t\right)$ by absorbing the first
term into its definition. This almost identity transformation of $\vec{c}_{1}\left(t\right)$
is a normal form transformation. It is the way the normal form method
of \cite{Soffer1999,Soffer2004} works : integration by parts , and
change of variables to absorb the boundary term. This approach gives
an equivalent normal form transformation to other methods \cite{Arnold1989,Ziane:2000:CRG,DeVille20081029,OMalleyJr20063}.

\end{rem}

We will use this method, adapted to our case to handle the multiple
iteration.

The above remark implies:

\begin{thm}\end{thm}

Under conditions of theorem \ref{-(Time-Averaging:thoerm1} above
,
\begin{equation}
\vec{c}\left(t\right)\equiv U_{0}\left(t\right)\tilde{U}^{-1}\tilde{c}_{1}\left(t\right)=U_{0}\left(t\right)\vec{c}_{1}\left(t\right)\label{eq:2.36-2}
\end{equation}
where
\begin{equation}
\tilde{U}f\equiv\left(I+i\beta\int_{T_{i}}^{t}A_{1}\left(s\right)ds\right)f\left(t\right).\label{eq:2.37-1-1}
\end{equation}
Then $\vec{c}_{1}$ satisfies
\begin{equation}
\sup_{T_{i}\leq t\leq T_{f}}\left\Vert \tilde{c}_{1}\left(t\right)-\tilde{c}_{1}\left(T_{i}\right)\right\Vert \leq C\left(\left|T_{f}-T_{i}\right|+\beta\left|T_{f}-T_{i}\right|^{2}\right)\beta^{\frac{3}{2}}.\label{eq:2.36-1}
\end{equation}

\begin{proof}
\begin{align}
i\frac{\partial}{\partial t}\tilde{c_{1}}\left(t\right) & =i\frac{\partial}{\partial t}\left(\tilde{U}\left(t\right)\vec{c_{1}}\left(t\right)\right)=-\beta A_{1}\left(t\right)\vec{c_{1}}\left(t\right)+\tilde{U}\left(t\right)\beta A_{1}\left(t\right)\vec{c_{1}}\left(t\right)\nonumber \\
= & -\beta A_{1}\left(t\right)\vec{c}_{1}\left(t\right)+\beta A_{1}\left(t\right)\vec{c_{1}}\left(t\right)+i\beta^{2}\left(\int_{T_{i}}^{t}A_{1}\left(s\right)ds\right)A_{1}\left(t\right)\vec{c}_{1}\left(t\right)\nonumber \\
= & i\beta^{2}\left(\int_{T_{i}}^{t}A_{1}\left(s\right)ds\right)A_{1}\left(t\right)\vec{c}_{1}\left(t\right)=\beta^{2}\left(\int_{T_{i}}^{t}A_{1}\left(s\right)ds\right)A_{1}\left(t\right)\tilde{U}^{-1}\tilde{c_{1}}\left(t\right).\nonumber \\
\end{align}
This equation can be obtained also from (\ref{eq:2.33-1}) replacing
$T_{f}$ by $t$, differentiating with respect to $t$ and using the
definition (\ref{eq:2.37-1-1}). Since $\tilde{U}$ is continuous,
the result (\ref{eq:2.36-1}) is found using the definition of $\tilde{c}_{1}\left(t\right)$,
namely (\ref{eq:2.36-2}) . This is detailed in section \ref{sec:Normal-form-transformation}.\end{proof}

We turn now to the higher levels of the hierarchy. First, we introduce
the global average of $A_{1}$ (as we did for $A=A_{0}$) and the
transformation $U_{1}$ corresponding to $U_{0}$. Let us define:
\begin{equation}
\bar{A}_{1}^{(n)}=\frac{1}{T_{0}}\int_{nT_{0}}^{(n+1)T_{0}}A_{1}(t)dt
\end{equation}
 Where $A_{1}(t)$ is given by (\ref{eq:2.17}). In analogy with (\ref{eq:2.6})
we define
\begin{equation}
\bar{A}_{1}^{g}(t)=\bar{A}_{1}^{(n)}\qquad\text{for }\qquad nT_{0}\leq t\leq(n+1)T_{0}.\label{eq:2.33}
\end{equation}
 To estimate this quantity we apply the integration by parts (\ref{eq:2.27})
to (\ref{eq:2.17}), resulting in
\begin{eqnarray}
\bar{A}_{1}^{(n)} & = & \frac{1}{T_{0}}\int_{nT_{0}}^{(n+1)T_{0}}U_{0}^{-1}(t)\left[A(t)-\bar{A}_{0}^{g}(t)\right]U_{0}(t)dt\nonumber \\
 & = & \frac{1}{T_{0}}\left[U_{0}^{-1}(t)\left(\int_{0}^{t}\left[A(t^{\prime})-\bar{A}_{0}^{g}(t^{\prime})\right]dt^{\prime}\right)U_{0}(t)\right]{}_{t=nT_{0}}^{(n+1)T_{0}}\label{eq:2.34}\\
 & - & \frac{1}{T_{0}}\int_{nT_{0}}^{(n+1)T_{0}}dt^{\prime}\left(\frac{d}{dt^{\prime}}U_{0}^{-1}(t^{\prime})\right)\left[\int_{0}^{t^{\prime}}(A(s)-\bar{A}_{0}^{g}(s))ds\right]U_{0}(t^{\prime})\nonumber \\
 & - & \frac{1}{T_{0}}\int_{nT_{0}}^{(n+1)T_{0}}U_{0}^{-1}(t^{\prime})\left[\int_{0}^{t^{\prime}}(A(s)-\bar{A}_{0}^{g}(s))ds\right]\frac{d}{dt}U_{0}(t^{\prime})dt^{\prime}.\nonumber
\end{eqnarray}
 The first term is zero since at time that is an integer multiple
of $T_{0}$, $\int_{0}^{t}\left[A(t)-\bar{A}_{0}^{g}(t)\right]$ vanishes
(see definition (\ref{eq:2.33-1})). From (\ref{eq:2.23}) and (\ref{eq:2.24})
we conclude that
\begin{equation}
\left|\int_{0}^{t}A(s)-\bar{A}_{0}^{g}(s)\right|<\frac{2\left\Vert A\right\Vert }{\sqrt{\beta}},
\end{equation}
The derivatives $dU_{0}^{-1}/dt$ and $dU_{0}(t)/dt$ are of order
$\beta$ due to (\ref{eq:2.10}) and (\ref{eq:2.11}), therefore the
second and third terms are of order $\sqrt{\beta}$. In conclusion,
\begin{equation}
\left|\bar{A}_{1}^{g}(t)\right|=O\left(\sqrt{\beta}\right).\label{eq:2.36}
\end{equation}
 Define in analogy to (\ref{eq:2.7}) and (\ref{eq:2.8}),
\begin{equation}
U_{1}(t)=e^{-i\beta\bar{A}_{1}^{(n)}(t-nT_{0})}\cdots e^{-i\beta\bar{A}_{1}^{(0)}T_{0}}\label{eq:2.37}
\end{equation}
 and
\begin{equation}
U_{1}^{-1}(t)=e^{i\beta\bar{A}_{1}^{(0)}T_{0}}\cdots e^{i\beta\bar{A}_{1}^{(n)}(t-nT_{0})}.
\end{equation}
 In analogy to (\ref{eq:2.10}) and (\ref{eq:2.11}) one finds
\begin{equation}
i\frac{\partial}{\partial t}U_{1}(t)=\beta\bar{A}_{1}^{g}(t)U_{1}(t)\label{eq:2.39}
\end{equation}
 and
\begin{equation}
i\frac{\partial}{\partial t}U_{1}^{-1}(t)=-\beta U_{1}^{-1}(t)\bar{A}_{1}^{g}(t).\label{eq:2.40}
\end{equation}
 In analogy to (\ref{eq:2.12}) define now
\begin{equation}
\vec{c}_{2}(t)=U_{1}^{-1}(t)\vec{c}_{1}(t)\label{eq:2.41}
\end{equation}
 and develop an equation analogous to (\ref{eq:2.16}), namely :
\begin{equation}
i\frac{\partial}{\partial t}\vec{c}_{1}(t)=\beta A_{1}(t)\vec{c}_{1}(t).
\end{equation}
 For this purpose we follow the steps (\ref{eq:2.13})-(\ref{eq:2.16}),
and obtain
\begin{eqnarray}
i\frac{\partial}{\partial t}\vec{c}_{2}(t) & = & \left[i\frac{\partial}{\partial t}U_{1}^{-1}(t)\right]\vec{c}_{1}(t)+U_{1}^{-1}(t)\left[i\frac{\partial}{\partial t}\vec{c}_{1}(t)\right]=\\
 & - & \beta U_{1}^{-1}\bar{A_{1}^{g}}(t)U_{1}\vec{c}_{2}(t)+\beta U_{1}^{-1}\bar{A}_{1}^{g}(t)U_{1}(t)\vec{c}_{2}(t),\nonumber
\end{eqnarray}
reducing to
\begin{equation}
i\frac{\partial}{\partial t}\vec{c}_{2}(t)=\beta U_{1}^{-1}(t)\left[A_{1}(t)-\bar{A}_{1}^{g}(t)\right]U_{1}(t)\vec{c}_{2}(t).\label{eq:2.44}
\end{equation}
 Define
\begin{equation}
A_{2}(t)=U_{1}^{-1}\left[A_{1}(t)-\bar{A}_{1}^{g}(t)\right]U_{1}(t),\label{eq:2.45}
\end{equation}
 leading to
\begin{equation}
i\frac{\partial}{\partial t}\vec{c}_{2}=\beta A_{2}(t)\vec{c}_{2}(t).\label{eq:2.46}
\end{equation}
 Equation (\ref{eq:2.46}) is similar in nature to (\ref{eq:2.16})
and is the following equation in the hierarchy. Estimates of $A_{2}$
can be performed in the same way as the estimates of $A_{1}.$ The
integral
\begin{equation}
I_{2}=\int_{0}^{t}dt'\, A_{2}\left(t'\right)
\end{equation}
 can be estimated in the same way as $I_{1}$ in (\ref{eq:I1}), namely,
\begin{eqnarray}
I_{2} & = & U_{1}^{-1}(t)\left(\int_{0}^{t}\left[A_{1}(t^{\prime})-\bar{A}_{1}^{g}(t^{\prime})\right]dt^{\prime}\right)U_{0}(t)\label{eq:I2}\\
 & - & \int_{0}^{t}\left(\beta U_{1}^{-1}\bar{A}_{1}^{g}\left(t^{\prime}\right)\right)\left(\int_{0}^{t^{\prime}}\left[A_{1}(s)-\bar{A}_{1}^{g}(s)\right]ds\right)U_{1}(t^{\prime})dt^{\prime}\nonumber \\
 & - & \int_{0}^{t}U_{1}^{-1}(t^{\prime})\left(\int_{0}^{t^{\prime}}\left[A_{1}(s)-\bar{A}_{1}^{g}(s)\right]d't\right)\left(-i\beta\bar{A}_{1}^{g}(t^{\prime})\right)U_{1}(t^{\prime})dt^{\prime}.\nonumber
\end{eqnarray}
 Due to (\ref{eq:2.36-1}), (\ref{eq:2.36}) holds.
\begin{equation}
\left|\bar{A}_{1}^{g}\right|=O\left(\sqrt{\beta}\right),\label{eq:2.49}
\end{equation}
 By reasoning similar to (\ref{eq:2.22})
\begin{equation}
\left|\int_{0}^{t}\left(A_{1}\left(t'\right)-\bar{A}_{1}^{g}\left(t'\right)\right)dt'\right|=O\left(\frac{1}{\sqrt{\beta}}\right),\label{eq:2.50}
\end{equation}
 and
\begin{equation}
I_{2}=O\left(\frac{1}{\sqrt{\beta}}\right)+O\left(\beta t\right).\label{eq:2.51}
\end{equation}
The difference between $I_{1}$and $I_{2}$ results of the fact that
$\bar{A_{1}^{g}}$ is of the order $\sqrt{\beta}$ while $\bar{A_{0}^{g}}$
is of order 1.

To estimate the magnitude of $\bar{A}_{2}^{\left(g\right)}$ defined
(in analogy to (\ref{eq:2.33-1})) as we note that
\begin{equation}
\bar{A}_{2}^{\left(g\right)}=\bar{A}_{2}^{\left(n\right)}\qquad\text{for}\qquad nT_{0}\leq t\leq\left(n+1\right)T_{0},
\end{equation}
 where
\begin{equation}
\bar{A}_{2}^{\left(n\right)}=\frac{1}{T_{0}}\int_{nT_{0}}^{\left(n+1\right)T_{0}}A_{2}\left(t'\right)dt'.\label{eq:2.53a}
\end{equation}
 We now repeat the calculation of (\ref{eq:2.34}), and use (\ref{eq:2.49})
(\ref{eq:2.40}) and (\ref{eq:2.41}) combined with (\ref{eq:2.36-1})
to estimate $dU_{1}/dt$ and $dU_{1}^{-1}/dt$. These derivatives
are of order $\beta^{3/2}$. Consequently,
\begin{equation}
\left|\bar{A}_{2}^{\left(g\right)}\right|=O\left(\beta\right).\label{eq:2.52}
\end{equation}
 We have constructed explicitly the first two stages of the hierarchy:
\begin{equation}
\vec{c}=U_{0}\vec{c}_{1}\label{eq:2.53}
\end{equation}
 and
\begin{equation}
\vec{c}_{1}=U_{1}\vec{c}_{2}.\label{eq:2.54}
\end{equation}
 The process can be continued further repeating (\ref{eq:2.17}) and
(\ref{eq:2.45}), defining
\begin{equation}
A_{n+1}(t)=U_{n}^{-1}(t)\left[A_{n}(t)-\bar{A}_{n}^{g}(t)\right]U_{n}(t).\label{eq:2.57}
\end{equation}
 The equations similar to (\ref{eq:2.16}) and (\ref{eq:2.46}) are
\begin{equation}
i\frac{\partial}{\partial t}\vec{c}_{n}\left(t\right)=\beta A_{n}\left(t\right)\vec{c}_{n}\left(t\right).
\end{equation}
 The results of the present section can be summarized in \begin{prop}
\label{prop:2}\end{prop}
\begin{enumerate}
\item
\[
i\frac{\partial}{\partial t}\vec{c}_{n}=\beta A_{n}(t)\vec{c}_{n}
\]

\item
\[
\vec{c}=U_{0}U_{1}\cdots U_{n-1}\vec{c}_{n}
\]

\item
\[
\left|\int_{0}^{t}A_{n}(s)-\bar{A}_{n}^{g}(s)ds\right|\leq\frac{\text{const}}{\sqrt{\beta}},
\]
since $\left|A\right|$ is bounded.
\item
\[
\left|A_{n}(s)-\bar{A}_{n}^{g}(s)\right|=O(1)
\]
leading to
\item
\[
\bar{A}_{n}^{g}=O\left(\beta^{n/2}\right)
\]
and
\item
\[
I_{n+1}=\left|\int_{0}^{t}U_{n}^{-1}(s)\left[A_{n}(s)-\bar{A}_{n}^{g}(s)\right]U_{n}(s)\right|\leq O\left(\frac{1}{\sqrt{\beta}}\right)+O\left(\beta^{n/2}t\right).
\]
Note that :
\end{enumerate}
\hphantom{}
\begin{enumerate}
\item is a generalization of (\ref{eq:2.16}) and (\ref{eq:2.46}).
\item results of a repeated application of the transformation like (\ref{eq:2.53})
and (\ref{eq:2.54}).
\item is a generalization of (\ref{eq:2.23}) and (\ref{eq:2.50}).
\item is a generalization of (\ref{eq:2.21}).
\item is a generalization of (\ref{eq:2.36-1}) and (\ref{eq:2.52}).
\item is a generalization of (\ref{eq:2.31 formally}) and (\ref{eq:2.51}).
\end{enumerate}
\begin{rem} From (\ref{eq:2.57}) we conclude that if the limiting
operator exists, it solves the following self-averaging equation:
\begin{equation}
\frac{1}{T_{0}}\int_{0}^{T_{0}}U_{\infty}^{-1}(s)\left[A_{\infty}(s)-Q(s)\right]U_{\infty}(s)ds=Q,
\end{equation}

\end{rem} where
\begin{equation}
U_{\infty}=\lim_{n\to\infty}U_{n}\qquad A_{\infty}=\lim_{n\to\infty}A_{n}\qquad Q=\lim_{n\to\infty}\bar{A}_{n}^{g}.
\end{equation}

\section{\label{sec:Normal-form-transformation}Normal form transformation}

In this section a normal form transformation will be applied to $\vec{c}_{2}\left(t\right)$
resulting in a new vector $\vec{c}_{2}^{U}\left(t\right)$. This vector
satisfies an equation similar to (\ref{eq:1.1}) but with $\beta$
replaced by $\beta^{3/2}$.

\begin{thm}\end{thm}

Under the conditions of Theorem 1 on $A\left(t\right)$, we have :
\begin{equation}
\vec{c}\left(t\right)=U_{0}U_{1}\tilde{U_{2}}^{-1}\vec{c}_{2}^{\,\, U}\left(t\right)
\end{equation}
satisfies
\begin{equation}
i\frac{\partial}{\partial t}c_{2}^{U}\left(t\right)=\beta^{\frac{3}{2}}\tilde{A}\left(t\right)\vec{c}_{2}^{\,\, U}\left(t\right),
\end{equation}
 where $U_{0}$ is the evolution operator of the partially averaged
dynamics (on time scales of order $\frac{1}{\sqrt{\beta}}$), and
$\tilde{U}_{2}$ is a unitary (normal form) almost identity transformation
while $\tilde{A}\left(t\right)$ satisfies the same conditions as
$A\left(t\right)$ of Theorem 1.

\begin{proof}

We start from (\ref{eq:2.44})
\begin{equation}
i\frac{\partial}{\partial t}\vec{c}_{2}=\beta U_{1}^{-1}(t)\left[A_{1}(t)-\bar{A}_{1}^{g}(t)\right]U_{1}(t)\vec{c}_{2}(t).
\end{equation}
 Integrating by parts, we get:
\begin{eqnarray}
\vec{c}_{2}(t)-\vec{c}_{2}(0) & = & -i\beta\int_{0}^{t}U_{1}^{-1}(t^{\prime})\left[A_{1}-\bar{A}_{1}^{g}(t^{\prime})\right]U_{1}(t^{\prime})\vec{c}_{2}(t^{\prime})dt^{\prime}\nonumber \\
 & = & -i\beta U_{1}^{-1}(t)\left(\int_{0}^{t}\left[A_{1}(t^{\prime})-\bar{A}_{1}^{g}(t^{\prime})\right]dt^{\prime}\right)U_{1}(t)\vec{c}_{2}(t)\\
 & + & i\beta\int_{0}^{t}dt^{\prime}\frac{dU_{1}^{-1}(t^{\prime})}{dt^{\prime}}\left(\int_{0}^{t^{\prime}}\left[A_{1}(s)-\bar{A}_{1}^{g}(s)\right]ds\right)U_{1}(t)\vec{c}_{2}(t)\nonumber \\
 & + & i\beta\int_{0}^{t}dt^{\prime}U_{1}^{-1}(t^{\prime})\left(\int_{0}^{t^{\prime}}\left[A_{1}(s)-\bar{A}_{1}^{g}(s)\right]ds\right)\frac{dU_{1}}{dt^{\prime}}(t^{\prime})\vec{c}_{2}(t^{\prime})\nonumber \\
 & + & i\beta\int_{0}^{t}dt^{\prime}U_{1}^{-1}(t^{\prime})\left(\int_{0}^{t^{\prime}}\left[A_{1}(s)-\bar{A}_{1}^{g}(s)\right]ds\right)U_{1}(t^{\prime})\frac{d}{dt^{\prime}}\vec{c}_{2}(t^{\prime}).\nonumber
\end{eqnarray}
 Using (\ref{eq:2.39}), (\ref{eq:2.40}) and (\ref{eq:2.46}) one
finds
\begin{eqnarray}
 &  & \vec{c}_{2}(t)-\vec{c}_{2}(0)+i\beta U_{1}^{-1}(t)\left(\int_{0}^{t}\left[A_{1}(t^{\prime})-\bar{A}_{1}^{g}(t^{\prime})\right]dt^{\prime}\right)U_{1}(t)\vec{c}_{2}(t)\label{eq:3.3}\\
 & = & i\beta\int_{0}^{t}dt^{\prime}\left(+i\beta U_{1}^{-1}(t)\right)\bar{A}_{1}^{g}(t)\left(\int_{0}^{t^{\prime}}\left[A_{1}(s)-\bar{A}_{1}^{g}(s)\right]ds\right)U_{1}(t^{\prime})\vec{c}_{2}(t^{\prime})\nonumber \\
 & + & i\beta\int_{0}^{t}dt^{\prime}U_{1}^{-1}(t^{\prime})\left(\int_{0}^{t^{\prime}}\left[A_{1}(s)-\bar{A}_{1}^{g}(s)\right]ds\right)(-i\beta)\bar{A}_{1}^{g}(t^{\prime})U_{1}(t^{\prime})\vec{c}_{2}(t^{\prime})\nonumber \\
 & + & i\beta\int_{0}^{t}dt^{\prime}U_{1}^{-1}(t^{\prime})\left(\int_{0}^{t^{\prime}}\left[A_{1}(s)-\bar{A}_{1}^{g}(s)\right]ds\right)U_{1}(t^{\prime})(-i\beta)A_{2}(t^{\prime})\vec{c}_{2}(t^{\prime})\nonumber \\
 & = & -\beta^{2}\int_{0}^{t}dt^{\prime}U_{1}^{-1}(t^{\prime})\bar{A}_{1}^{g}(t^{\prime})\int_{0}^{t^{\prime}}\left[A_{1}(s)-\bar{A}_{1}^{g}(s)\right]U_{1}(t^{\prime})\vec{c}_{2}(t^{\prime})\nonumber \\
 & + & \beta^{2}\int_{0}^{t}dt^{\prime}U_{1}^{-1}(t^{\prime})\left(\int_{0}^{t^{\prime}}\left[A_{1}(s)-\bar{A}_{1}^{g}(s)\right]ds\right)\bar{A}_{1}^{g}(t^{\prime})U_{1}(t^{\prime})\vec{c}_{2}(t^{\prime})\nonumber \\
 & + & \beta^{2}\int_{0}^{t}dt^{\prime}U_{1}^{-1}(t^{\prime})\left(\int_{0}^{t^{\prime}}\left[A_{1}(s)-\bar{A}_{1}^{g}(s)\right]ds\right)U_{1}(t^{\prime})A_{2}(t^{\prime})\vec{c}_{2}(t^{\prime}).\nonumber
\end{eqnarray}
 By (\ref{eq:2.50}) and Proposition \ref{prop:2},

\begin{equation}
\left|\int_{0}^{t}(A_{1}(s)-\bar{A}_{1}^{g}(s)ds\right|<O\left(\frac{1}{\sqrt{\beta}}\right)\label{eq:3.4}
\end{equation}
 and by \ref{eq:2.36-1}
\begin{equation}
|A_{2}|\sim O(1)\quad\text{and}\quad\bar{A}_{1}^{g}\sim O(\sqrt{\beta}),\label{eq:3.5}
\end{equation}
 while $\bar{A}_{2}^{g}=O\left(\beta\right)$ by (\ref{eq:2.52}).
The first two terms on the RHS of (\ref{eq:3.3}) are of order $\beta^{2}$
and the last is term of order $\beta^{3/2}$. We turn now to perform
the normal form transformation. For this purpose we rewrite (\ref{eq:3.3})
in the form
\begin{equation}
\vec{c}_{2}(t)+i\beta U_{1}^{-1}(t)\left(\int_{0}^{t}\left[A_{1}(t^{\prime})-\bar{A}_{1}^{g}(t^{\prime})\right]dt^{\prime}\right)U_{1}\left(t\right)\vec{c}_{2}(t)-\vec{c}_{2}\left(0\right)=\text{RHS}.\label{eq:3.6}
\end{equation}
 Then we define,
\begin{equation}
\vec{c}_{2}^{U}\equiv\tilde{U}\vec{c}_{2}\label{eq:3.7}
\end{equation}
 where
\begin{equation}
\tilde{U}=1+i\beta U_{1}^{-1}(t)\left(\int_{0}^{t}\left[A_{1}(t^{\prime})-\bar{A_{1}^{g}}(t^{\prime})\right]dt^{\prime}\right)U_{1}(t).\label{eq:3.8}
\end{equation}
 With this definition (\ref{eq:3.6}) takes the form
\begin{equation}
\vec{c}_{2}^{U}\left(t\right)=\tilde{U}\vec{c}_{2}\left(t\right)-\vec{c}_{2}\left(0\right)=\text{RHS},
\end{equation}
 and differentiation of this equation with respect to time yields,
\begin{equation}
i\frac{\partial}{\partial t}\vec{c}_{2}^{U}=\beta^{3/2}\tilde{A}(t)\vec{c}_{2}^{U}\label{eq:3.10}
\end{equation}
 where
\begin{eqnarray}
\tilde{A}(t) & = & \beta^{1/2}\left[\left(-U_{1}^{-1}(t)\bar{A}_{1}^{g}(t)\right)\left(\int_{0}^{t}\left[A_{1}(s)-\bar{A}_{1}^{g}(s)\right]ds\right)U_{1}(t)\right.\nonumber \\
 & + & U_{1}^{-1}(t)\left(\int_{0}^{t}\left[A_{1}(s)-\bar{A}_{1}^{g}(s)\right]ds\right)\bar{A}_{1}^{g}(t)U_{1}(t)\nonumber \\
 & + & \left.U_{1}^{-1}(t)\int_{0}^{t}\left[A_{1}(s)-\bar{A_{1}^{g}}(s)\right]U_{1}(t)A_{2}(t)\right]\tilde{U}^{-1}(t).
\end{eqnarray}
 By (\ref{eq:3.4}) and (\ref{eq:3.5}) and $\left\Vert \tilde{A}(t)\right\Vert =O\left(1\right)$.
From the definition of (\ref{eq:3.8}), it is clear that $\tilde{U}$
is also of order $O\left(1\right)$. It should be noted, however,
that the normal form transformation above, is not unitary. Therefore,
there is more than one way, to make it a unitary operator, by adding
a correction to $\tilde{A}$, of higher order in $\beta$. We do that
by redefining $\tilde{U}$:
\[
\widetilde{U}=1+i\beta U_{1}^{-1}B(t)U_{1}\rightarrow\mathcal{T}\, e^{i\beta\int_{0}^{t}U_{1}^{-1}(s)B(s)U_{1}(s)ds}
\]
 Where $B(t)=\int_{0}^{t}\left[A_{1}(t^{\prime})-\bar{A}_{1}^{g}(t^{\prime})\right]dt^{\prime}$
is self-adjoint. By choosing a unitary $\tilde{U}$ to generate the
normal form transformation, we ensure that the $L^{2}$ norm does
not change in time. Therefore, by Proposition \ref{prop:1}, the overall
flow, generated by a product of unitary flows, is given in terms of
a time dependent self-adjoint generator. \end{proof}

Equation (\ref{eq:3.10}) is similar to (\ref{eq:1.1}) with the replacement,

\begin{eqnarray}
\vec{c} & \rightarrow & \vec{c}_{2}^{U}\nonumber \\
A & \rightarrow & \tilde{A}\nonumber \\
\beta & \rightarrow & \beta^{3/2}.
\end{eqnarray}
 This reduces $\beta$ effectively and increases the averaging time
$T_{0}$.

\begin{cor} \end{cor}

We can continue the process and find $\vec{c}_{n}\left(t\right)$,
such that
\begin{equation}
\vec{c}=V_{1}\cdots V_{n-1}\vec{c}_{n}.
\end{equation}
 Using (\ref{eq:2.53}) and (\ref{eq:2.54})
\begin{equation}
\vec{c}=U_{1}U_{2}\vec{c}_{2},
\end{equation}
 and using (\ref{eq:3.7})
\begin{eqnarray}
\vec{c}_{2} & = & \tilde{U}^{-1}\vec{c}_{2}^{U}.
\end{eqnarray}
 If one defines,
\begin{eqnarray}
V_{1} & = & U_{1}U_{2}\tilde{U}^{-1},
\end{eqnarray}
 then
\begin{equation}
\vec{c}=V_{1}\vec{c}_{2}^{U}
\end{equation}
 Repeating these steps results in a process given in (3.12) after
an appropriate renumbering of the $\vec{c}_{n}$. The $n-$th step
of this process (\ref{eq:1.7}), results in the replacement
\begin{eqnarray}
V_{1} & \rightarrow & V_{n}\\
U_{1} & \rightarrow & U_{1,n}\nonumber \\
U_{2} & \rightarrow & U_{2,n}\nonumber \\
\tilde{U} & \to & \tilde{U}_{n}.\nonumber
\end{eqnarray}

\section{\label{sec:Quasiperiodic-System}Quasiperiodic System}

A crucial difficulty with perturbative schemes, is that these often
produce terms with arbitrarily slow frequencies. Such terms lead,
upon integration to \emph{small divisors}. Here, we will treat such
a problem, and show, how to get the large time behavior of the system,
to any order, via multi-scale time averaging. Our main example is
the following: \begin{exmp} Let $M_{j}$ be $N\times N$ matrices
with norm bounded by $\|M_{j}\|\leq j^{-1-\delta},$ for all $j>0.$
Assume,
\begin{equation}
A=\sum_{j=1}^{\infty}(M_{j}e^{i\omega_{j}t}+h.c)\label{eq:4.1}
\end{equation}
 where h.c stands for hermitian conjugate. Note that norm convergence
of the sum is assured by our assumption on the norms of $M_{j}$.
The case of interest is when $\omega_{j}\rightarrow0$ as $j\rightarrow\infty.$
\end{exmp}

To see how such systems arise , consider the problem of solving a
system of the type
\begin{align}
i\frac{d\vec{c}}{dt} & =A_{0}\vec{c}+\beta F\left(\vec{c}\right)\vec{c}\,\,\,,\,\,\vec{c}\in\mathbb{C^{N}}\\
 & \vec{c}\left(t=0\right)=\vec{c}_{0}
\end{align}
 where ~ $A_{0}$ is a time Independent and hermitian matrix with
eigenvalues $E_{1},E_{2}.....E_{N}$; $F\left(\vec{c}\right)$ is
an $\left(N\times N\right)$ matrix (depending on $\vec{c}$ ) which
we assume is also hermitian for all $\vec{c}$.

Then we try to solve this problem by iterations : solving the first
order system gives
\begin{equation}
\vec{c}_{0}\left(t\right)=e^{-iA_{0}t}\vec{c_{0}}
\end{equation}
which is a sum of oscillations with frequencies $E_{1},....,E_{N},$
Iterating once we get after defining
\begin{equation}
\vec{c}_{1}\left(t\right)=e^{+iA_{0}t}\vec{c_{0}}
\end{equation}
 that $\vec{c}_{1}\left(t\right)$satisfies an equation of the form
\begin{align}
i\frac{d}{dt}\vec{c}_{1}\left(t\right) & =\beta e^{iA_{0}t}F\left(\vec{c}\right)e^{-iA_{0}t}\vec{c}_{1}\left(t\right)\\
= & \beta e^{iA_{0}t}F\left(\vec{c}_{1}\right)e^{-iA_{0}t}\vec{c}_{1}\left(t\right)+.....
\end{align}
In general the frequencies of $e^{iA_{0}t}$ combine with the frequencies
of $\vec{F}\left(\vec{c}_{1}\left(t\right)\right)$ and generate many
new frequencies .

For $\vec{F}$ depending polynomially on $\vec{c}\left(t\right)$,
they are of the general form
\begin{equation}
\omega_{M}=\sum_{i=1}^{M}n_{i}E_{i},\,\,\, n_{i}\,\, integers.
\end{equation}
 Such sums can add to zero (secular terms ) or add to a very small
number .Formally integrating the equation with such a term leads to
expression of the type
\begin{equation}
\sim\frac{e^{i\omega_{M}t}-1}{\omega_{M}}
\end{equation}
which is the small divisor problem , when $\omega_{M}$ is zero or
small.

It is clear that in these kind of iterations scheme , each linear
approximation is of the type $\left(4.1\right)$.

In order to apply our method to the linear system we have to split
$A\left(t\right)$ of (\ref{eq:4.1}) as
\begin{equation}
A=A^{>}+A^{<},\label{eq:4.2}
\end{equation}
 and to treat separately the two parts. For this purpose we introduce
the following two Lemmas. \begin{lem} \label{lem:2}Let
\begin{equation}
A^{>}=\sum_{j=j_{0}}^{\infty}(M_{j}e^{i\omega_{j}t}+h.c)
\end{equation}
 with $j_{0}$ sufficiently large, so that
\begin{equation}
\omega_{j}T_{0}\ll1,\label{eq:4.4}
\end{equation}
 with $T_{0}=\beta^{-1/2}$ and $\beta$ small, and
\begin{equation}
\|M_{j}\|T_{0}\leq\frac{1}{j^{1+\delta}}\qquad\text{,\qquad}\delta>0,\label{eq:4.13-1}
\end{equation}
 for all $j\geq j_{0}$. Then,
\begin{equation}
\left|\int_{0}^{t}(A^{>}(t^{\prime})-\bar{A}_{0}^{>g}(t^{\prime}))dt^{\prime}\right|\leq O(1)
\end{equation}
 \end{lem} \begin{proof} Start from ,
\begin{equation}
A^{>}=\sum_{j=j_{0}}^{\infty}M_{j}e^{i\omega_{j}t}+h.c,
\end{equation}
 and choose a term and denote it
\begin{equation}
B(t)=M_{j}e^{i\omega_{j}t}
\end{equation}

\begin{equation}
\bar{B}_{0}^{(0)}=\frac{1}{T_{0}}\int_{0}^{T_{0}}B(t)dt=\frac{e^{i\omega_{j}T_{0}}-1}{i\omega_{j}T_{0}}M_{j}
\end{equation}
 and
\begin{eqnarray}
\bar{B}_{0}^{(n)} & = & \frac{M_{j}}{T_{0}}\int_{nT_{0}}^{(n+1)T_{0}}e^{i\omega_{j}t}dt=M_{j}\frac{e^{i(n+1)T_{0}\omega_{j}}-e^{inT_{0}\omega_{j}}}{i\omega_{j}T_{0}}\nonumber \\
 & = & M_{j}e^{inT_{0}\omega_{j}}\left(\frac{e^{iT_{0}\omega_{j}}-1}{i\omega_{j}T_{0}}\right).\label{eq:4.10}
\end{eqnarray}
 Define
\begin{equation}
\bar{B}^{g}(t)=\bar{B}_{0}^{(n)}
\end{equation}
 for
\begin{equation}
nT_{0}\leq t<(n+1)T_{0}.
\end{equation}
 Then,
\begin{eqnarray*}
\int_{0}^{t}dt^{\prime}\left[B(t^{\prime})-\bar{B}_{0}^{g}(t^{\prime})\right] & = & \int_{0}^{T_{0}}\left[B(t^{\prime})-\bar{B}_{0}^{\left(0\right)}(t^{\prime})\right]dt^{\prime}+\cdots+\int_{jT_{0}}^{(j+1)T_{0}}\left[B(t^{\prime})-\bar{B}^{(j)}(t^{\prime})\right]dt^{\prime}\\
+\cdots+\int_{nT_{0}}^{t}\left[B(t^{\prime})-\bar{B}^{\left(n\right)}\left(t^{\prime}\right)\right]dt^{\prime} & = & \int_{nT_{0}}^{t}\left[B(t^{\prime})-\bar{B}_{0}^{g}\left(t^{\prime}\right)\right]dt^{\prime}
\end{eqnarray*}
\begin{equation}
\label{eq:4.13}
\end{equation}
 where
\begin{equation}
nT_{0}\leq t<(n+1)T_{0}.
\end{equation}
 Since,
\begin{equation}
\frac{1}{T_{0}}\int_{jT_{0}}^{(j+1)T_{0}}B(t')dt'=\bar{B}_{0}^{(j)}.
\end{equation}
 one finds,
\begin{equation}
\int_{jT_{0}}^{(j+1)T_{0}}\left[B(t^{\prime})-\bar{B}_{0}^{g}(t^{\prime})\right]dt'=0.
\end{equation}
 Using (\ref{eq:4.10}) and (\ref{eq:4.13-1}) one finds,
\begin{eqnarray}
\int_{0}^{t}dt^{\prime}\left[B(t^{\prime})-\bar{B}_{0}^{g}(t^{\prime})\right] & = & M_{j}\frac{e^{i\omega_{j}t}-e^{i\omega_{j}nT_{0}}}{i\omega_{j}}-M_{j}\frac{\left(t-nT_{0}\right)e^{inT_{0}\omega_{j}}\left(e^{i\omega_{j}T_{0}}-1\right)}{i\omega_{j}T_{0}}\nonumber \\
 & = & \frac{M_{j}}{i\omega_{j}}e^{i\omega_{j}nT_{0}}\left[\left(e^{i\omega_{j}(t-nT_{0})}-1\right)-\left(t-nT_{0}\right)\left(\frac{e^{i\omega_{j}T_{0}}-1}{T_{0}}\right)\right].
\end{eqnarray}
 For,
\begin{equation}
\omega_{j}T_{0}\ll1,
\end{equation}
 where $T_{0}=\beta^{-1/2}$. Using the fact that $\left|t-nT_{0}\right|<T_{0}$
we get:
\begin{equation}
\left|\int_{0}^{t}\left[B(t^{\prime})-\bar{B}_{0}^{g}(t^{\prime})\right]dt^{\prime}\right|\leq||M_{j}||\cdot2T_{0}.
\end{equation}
 Assuming (\ref{eq:4.13-1}) ,
\begin{equation}
\left\Vert M_{j}\right\Vert T_{0}\leq\frac{1}{j^{i+\delta}},
\end{equation}
 and summing the various contributions in the sum for $A^{>}$, one
finds,
\begin{eqnarray}
\left|\int_{0}^{t}\left[A^{>}(t^{\prime})-\bar{A}_{0}^{>g}(t^{\prime})\right]dt^{\prime})\right| & \leq & \sum_{j=j_{0}}^{\infty}2||M_{j}||T_{0}\leq\text{const}<\infty
\end{eqnarray}
 or
\begin{equation}
\left|\int_{0}^{t}dt'\left(A^{>}(t')-\bar{A}_{0}^{>g}(t')\right)\right|\leq\ensuremath{O(1)}.\label{eq:4.21}
\end{equation}
 \end{proof} \begin{rem} Note the difference between (\ref{eq:4.21})
and (\ref{eq:2.50}). This is due to the fact that we can use the
special form of $A\left(t\right)$, to separate the low frequency
terms from the ``order 1'' frequency terms. \end{rem} \begin{lem}
\label{lem:3}Let
\begin{equation}
A^{<}=\sum_{j=1}^{j_{0}-1}M_{j}e^{i\omega_{j}t}+c.c
\end{equation}
 and let
\begin{equation}
\frac{\|M_{j}\|}{\omega_{j}}\leq\frac{\mathrm{const}}{j^{1+\delta}},\,\,\,\,\,1\leq j<j_{0}
\end{equation}
 and $T_{0}=\beta^{-1/2}$. Then,
\begin{equation}
\left\Vert \frac{1}{T_{0}}\int_{0}^{t}A^{<}(s)ds\right\Vert \leq\mathrm{const}\beta^{1/2}.
\end{equation}
 \end{lem} \begin{proof}
\begin{equation}
A^{<}=\sum_{j=1}^{j_{0}-1}M_{j}e^{i\omega_{j}t},
\end{equation}
 we note that,
\begin{equation}
\int_{0}^{t}e^{i\omega_{j}t}=\frac{1}{i\omega_{j}}\left(e^{i\omega_{j}t}-1\right)
\end{equation}
 therefore,
\begin{equation}
\left|\int_{0}^{t}e^{i\omega_{j}t}\right|\leq\frac{2}{\omega_{j}}
\end{equation}
 Therefore, the relevant sum is bounded as,
\begin{equation}
\left|\int_{0}^{t}A^{>}\left(t'\right)dt'\right|\leq\sum_{j=1}^{j_{0}-1}\frac{2\left\Vert M_{j}\right\Vert }{\omega_{j}}\leq\text{const}.
\end{equation}
 Therefore,
\begin{equation}
\frac{1}{T_{0}}\left|\int_{0}^{t}A^{>}\left(t'\right)dt'\right|<\mathrm{const}\beta^{1/2}.
\end{equation}
 \end{proof} \begin{rem} The partition (\ref{eq:4.2}) is always
possible since if (\ref{eq:4.4}) is not satisfied then $\omega_{j}T_{0}$
larger then some constant that is much smaller than unity. Therefore,
all terms that do not belong to $A^{>}$ belong to $A^{<}$ with the
appropriate choice of the constant. This decomposition depends on
the choice of $T_{0}$. \end{rem} \begin{thm} The global average
$\bar{A}_{2}^{g}$ defined by (\ref{eq:2.52}) is bounded by
\begin{equation}
\left|\bar{A}_{2}^{g}\right|\leq O\left(\beta^{2}\right).
\end{equation}
 \end{thm} \begin{proof} First we estimate $\bar{A}_{1}^{g}$. For
this we define
\begin{eqnarray}
\bar{A}_{1}^{\left(0\right)} & = & \frac{1}{T_{0}}\int_{0}^{T_{0}}U_{0}^{-1}\left[A-\bar{A}{}_{0}^{g}\right]U_{0}\\
 & = & \frac{1}{T_{0}}\left[U_{0}^{-1}(t)\left(\int_{0}^{t}\left[A(s)-\bar{A}_{1}^{g}(s)\right]ds\right)U_{0}(t)\right]_{t=0}^{t=T_{0}}\nonumber \\
 & - & \frac{1}{T_{0}}\int_{0}^{T_{0}}\frac{dU_{0}^{-1}(t)}{dt}\left(\int_{0}^{t}\left[A(s)-\bar{A}_{0}^{\left(0\right)}\right]ds\, U_{0}(t)\right)dt\nonumber \\
 & - & \frac{1}{T_{0}}\int_{0}^{T_{0}}\int U_{0}^{-1}\left(\int_{0}^{t}\left[A(s)-\bar{A}_{0}^{\left(0\right)}\right]ds\right)\frac{dU_{0}}{dt}dt\nonumber
\end{eqnarray}
 From the definition of $\bar{A}_{0}^{\left(0\right)}(s)$,
\begin{equation}
\int_{0}^{T_{0}}ds\left(A(s)-\bar{A}_{0}^{\left(0\right)}\right)=0,\label{eq:4.33}
\end{equation}
 The contribution from the regime where Lemma \ref{lem:2} holds is,
\begin{equation}
\left|\int_{0}^{t}ds\left[A(s)-\bar{A}_{0}^{g}(s)\right]\right|\leq O(1)\label{eq:4.34}
\end{equation}

by (\ref{eq:4.21}). The contribution from the region where Lemma
\ref{lem:3} is relevant, is
\begin{equation}
\sum_{j=0}^{j_{0}-1}\frac{\left\Vert M_{j}\right\Vert }{\omega_{j}}\left[e^{i\omega_{j}\left(t-nT_{0}\right)}-e^{inT_{0}}-\frac{\left(t-nT_{0}\right)}{T_{0}}\left(e^{i\left(n+1\right)T_{0}\omega_{j}}-e^{inT_{0}\omega_{j}}\right)\right]\leq4\sum_{j=0}^{j_{0}-1}\frac{\left\Vert M_{j}\right\Vert }{\omega_{j}},
\end{equation}
 and an equality similar to (\ref{eq:4.34}) holds. From the general
theory (\ref{eq:2.10}) and (\ref{eq:2.11}),
\begin{equation}
\left|\frac{dU_{0}}{dt}\right|\leq O(\beta)\text{,}\qquad\left|\frac{dU_{0}^{-1}}{dt}\right|\leq O(\beta)
\end{equation}
 Therefore, combined with (\ref{eq:2.36-1})
\begin{equation}
|\bar{A}_{1}^{g}|\leq O(\beta)
\end{equation}
 by (\ref{eq:2.45}) and (\ref{eq:2.53a})
\begin{eqnarray}
\bar{A}_{2}^{\left(0\right)} & = & \frac{1}{T_{0}}\int_{0}^{T_{0}}U_{1}^{\prime-1}(t^{\prime})\left[A_{1}(t^{\prime})-\bar{A}_{1}^{g}(t^{\prime})\right]U_{1}(t^{\prime})dt^{\prime}\nonumber \\
 & = & \frac{1}{T_{0}}U_{1}^{-1}(t)\left(\int_{0}^{T_{0}}\left[A_{1}(s)-\bar{A}_{1}^{g}\right]ds\right)U_{1}(t)\nonumber \\
 & - & \frac{1}{T_{0}}\int_{0}^{T_{0}}\frac{dU_{1}^{-1}}{ds}\left(\int_{0}^{s}\left[A_{1}(s^{\prime})-\bar{A}_{1}^{g}(s^{\prime})\right]ds^{\prime}\right)U_{1}(s)ds\nonumber \\
 & - & \frac{1}{T_{0}}\int_{0}^{T_{0}}U_{1}^{-1}\left(\int_{0}^{s}\left[A_{1}(s^{\prime})-\bar{A}_{1}^{g}(s^{\prime})\right]ds^{\prime}\right)\frac{d}{ds}U_{1}(s)ds.\label{eq:4.38}
\end{eqnarray}
 By (\ref{eq:2.50}) of the general theory,
\begin{equation}
\frac{1}{T_{0}}\left|\int_{0}^{T_{0}}\left[A_{1}(s)-\bar{A}_{1}^{g}(s)\right]ds\right|\leq O\left(1\right).\label{eq:4.39}
\end{equation}
 From the definition of $\bar{A}_{1}^{g}$ one finds $\int_{0}^{T_{0}}\left(A_{1}\left(s\right)-A_{1}^{g}\left(s\right)\right)ds=0$.
Using (\ref{eq:4.38}) combined with (\ref{eq:2.39}) and (\ref{eq:2.40})
we find
\begin{equation}
\left|\frac{dU_{1}}{dt}\right|\leq O\left(\beta^{2}\right)\qquad\left|\frac{dU_{1}^{-1}}{dt}\right|\leq O\left(\beta^{2}\right),
\end{equation}
 then (\ref{eq:4.38}) combined with (\ref{eq:4.39}) and (\ref{eq:2.40})
leads to the bound,
\begin{equation}
\left|\bar{A}_{2}^{\left(0\right)}\right|\leq O\left(\beta^{2}\right).\label{eq:4.40}
\end{equation}
The same bound holds for all $\bar{A_{2}}^{\left(n\right)}$. But
$\bar{A}_{2}^{g}$ in each interval of length $T_{0}$ is equal to
one of the $\bar{A}_{2}^{\left(n\right)}$, each satisfying (\ref{eq:4.40}),
therefore it also satisfies (\ref{eq:4.40}), therefore
\begin{equation}
\left|\bar{A}_{2}^{g}\right|\leq O\left(\beta^{2}\right).
\end{equation}
\end{proof}

\begin{rem} This bound is better than the one of the general theory
(\ref{eq:2.52}). \end{rem} Part of this work was done at the Max-Planck
Institut in Dresden. A.S. thanks the Institute for partial support.
The work of A.S. is partially supported by NSF grant DMS-1201394.
This work was partly supported also by the Israel Science Foundation
(ISF) and by the US-Israel Binational Science Foundation (BSF), by
the Minerva Center of Nonlinear Physics of Complex Systems, by the
Shlomo Kaplansky academic chair and by the Fund for promotion of research
at the Technion.


\begin{thebibliography}{10}

\bibitem{Ziane:2000:CRG}
M.~Ziane.
\newblock On a certain renormalization group method.
\newblock 41(5):3290--3299, May 2000.

\bibitem{DeVille20081029}
R.E.~Lee DeVille, A.~Harkin, M.~Holzer, K.~Josic, and T.~J. Kaper.
\newblock Analysis of a renormalization group method and normal form theory for
  perturbed ordinary differential equations.
\newblock {\em Physica D: Nonlinear Phenomena}, 237(8):1029 -- 1052, 2008.

\bibitem{OMalleyJr20063}
R.~E.~OMalley Jr. and D.~B. Williams.
\newblock Deriving amplitude equations for weakly-nonlinear oscillators and
  their generalizations.
\newblock {\em Journal of Computational and Applied Mathematics}, 190(1).

\bibitem{PhysRevE.54.376}
L.~Chen, N.~Goldenfeld, and Y.~Oono.
\newblock Renormalization group and singular perturbations: Multiple scales,
  boundary layers, and reductive perturbation theory.
\newblock {\em Phys. Rev. E}, 54:376--394, Jul 1996.

\bibitem{Soffer1999}
A.~Soffer and M.~I. Weinstein.
\newblock Resonances, radiation damping and instabilitym in {H}amiltonian
  nonlinear wave equations.
\newblock {\em Invent. Math.}, 136(1):9--74, 1999.

\bibitem{Soffer2004}
A.Soffer and M.~I.Weinstein.
\newblock Selection of the ground state for nonlinear {S}chr\"{o}dinger
  equations.
\newblock {\em Rev. Math. Phys.}, 16(8):977--1071, 2004.

\bibitem{Arnold1989}
V.~I. Arnold.
\newblock {\em Mathematical Methods of Classical Mechanics}.
\newblock Springer, 2nd edition, 1989.

\bibitem{2009arXiv0912.3725B}
A.~{Bounemoura} and L.~{Niederman}.
\newblock {Generic Nekhoroshev theory without small divisors}.
\newblock {\em ArXiv 0912.3725}, December 2010.

\bibitem{1367-2630-12-6-063035}
Y~.Krivolapov S~.Fishman and A~.Soffer.
\newblock A numerical and symbolical approximation of the nonlinear anderson
  model.
\newblock {\em New Journal of Physics}, 12(6):063035, 2010.

\bibitem{Fishman2008a}
S.~Fishman, Y.~Krivolapov, and A.~Soffer.
\newblock On the problem of dynamical localization in the nonlinear
  {S}chr\"{o}dinger equation with a random potential.
\newblock {\em J. Stat. Phys.}, 131(5):843--865, 2008.

\bibitem{Fishman2009a}
S.~Fishman, Y.~Krivolapov, and A.~Soffer.
\newblock Perturbation theory for the nonlinear {S}chr\"{o}dinger equation with
  a random potential.
\newblock {\em Nonlinearity}, 22:2861--2887, 2009.

\bibitem{0951-7715-25-4-R53}
S.Fishman Y.Krivolapov and A.Soffer.
\newblock The nonlinear {S}chr\"{o}dinger equation with a random potential:
  results and puzzles.
\newblock {\em Nonlinearity}, 25(4):R53, 2012.

\bibitem{Sanders2010}
J.~A. Sanders, F.~Verhulst, and J.~Murdock.
\newblock {\em Averaging Methods in Nonlinear Dynamical Systems}.
\newblock Springer, 2nd ed. edition, 2010.

\bibitem{Krol1991}
M.~S. Krol.
\newblock On the averaging method in nearly time-periodic advection-diffusion
  problems.
\newblock {\em SIAM J. Appl. Math.}, 51:1622--1637, 1991.

\bibitem{0951-7715-22-12-004}
S.~Fishman, Y.~Krivolapov, and A~Soffer.
\newblock Perturbation theory for the nonlinear {S}chr\"{o}dinger equation with
  a random potential.
\newblock {\em Nonlinearity}, 22(12):2861, 2009.

\bibitem{reed1980methods}
M.~Reed and B.~Simon.
\newblock {\em Methods of Modern Mathematical Physics: Functional analysis}.
\newblock Number v. 1 in Methods of Modern Mathematical Physics Series.
  Academic Press, 1980.

\end{thebibliography}

\end{document}